\documentclass[conference]{IEEEtran}
\IEEEoverridecommandlockouts
\ifCLASSINFOpdf

\else

\fi
\usepackage{amsmath}
\usepackage{makeidx}  
\usepackage{algorithm}
\usepackage{algorithmic}
\usepackage{graphicx}
\usepackage{subfigure}
\usepackage{epstopdf}
\usepackage{bm}
\usepackage{cite}
\usepackage{stfloats}

\usepackage{amssymb}
\usepackage{graphicx}

\usepackage{url}
\usepackage{bbm}

\usepackage{tabularx,booktabs}
\newcolumntype{C}{>{\centering\arraybackslash}X} 
\usepackage{lipsum}

\usepackage{makecell} 

\usepackage{graphicx}
\usepackage{color}
\newtheorem{thm}{Theorem}

\newtheorem{proof}{proof}

\newtheorem{appro}{Approximation}
\usepackage{stfloats}

\begin{document}

\title{  Analysis and Optimization of Random Cache in Multi-Antenna HetNets with Interference Nulling
}

\author{\IEEEauthorblockN{Kangda Zhi\textsuperscript{*},  Guangji Chen\textsuperscript{*}, Ling~Qiu\textsuperscript{*},  Xiaowen Liang\textsuperscript{*}, and Chenhao Ren\textsuperscript{$\dagger$}}
       \thanks{This work was supported by National Natural Science Foundation of China under Grant No.61672484. (Corresponding author: Ling Qiu.)}
\IEEEauthorblockA{\textsuperscript{*}School of Information Science and Technology, University of Science and Technology of China, Hefei, China.\\
\textsuperscript{$\dagger$} Department of Electrical and Computer Engineering, University of California, Davis, USA
\\ Email: (kdzhi, chen0404)@mail.ustc.edu.cn, (lqiu, lxw)@ustc.edu.cn, cheren@ucdavis.edu}}




\maketitle

\begin{abstract}
From the perspective of statistical performance, this paper presents a framework for the per-user throughput analysis in random cache based multi-antenna heterogeneous networks (HetNets) with user-centric inter-cell interference nulling (IN). Using tools from stochastic geometry, an explicit expression for the per-user throughput is derived. Based on the analytical results, the optimal cache probabilities for maximizing the per-user throughput are analyzed. Theoretical analysis and numerical results reveal that the optimal random cache under interference nulling fully harvests the file diversity gain (FDG) and achieves a promising per-user throughput.

\end{abstract}


%

\IEEEpeerreviewmaketitle

\section{Introduction}

Heterogeneous networks (HetNets) and multi-antenna have been recognized as two key technologies to meet the predicted throughput requirements in 5G  networks\cite{andrews2014will}. However, rapidly growing traffic makes the backhaul rate a major bottleneck in implementing these technologies in practical systems. To address these problems, caching has been recognized as a promising technology to reduce the backhaul traffic and achieve low latency by avoiding duplicate delivery of popular content\cite{li2018survey}.

Recently, contrasted to traditional grid model, a random network model based on Poisson point processes (PPPs) has been widely applied in  cache-enabled small cell networks to characterize the irregularity and randomness of base station (BS) locations\cite{bacstug2015cache,wu2018content,tamoor2016caching,7417458}. These works proposed and analyzed some traditional cache placement schemes, such as the most popular cache (MPC) scheme \cite{bacstug2015cache}\cite{wu2018content},  the uniform cache (UC) scheme\cite{tamoor2016caching} and the  i.i.d. cache (IIDC) scheme\cite{7417458}. However, these  cache schemes \cite{bacstug2015cache,wu2018content,tamoor2016caching,7417458} cannot sufficiently exploit the finite storage capacity, and may not yield optimal network performance. In contrast, random cache (RC) scheme has been proved to achieve better performance by leveraging file diversity and file popularity. Recent contributions have considered the analysis and optimization of various performance metrics  in random cache based small cell networks, e.g.  the throughput\cite{chen2017probabilistic}, the hit probability\cite{wen2017cache}, and the successful transmission probability (STP)\cite{cui2016analysis}. {Note that \cite{chen2017probabilistic,wen2017cache,cui2016analysis} focus on the scenarios without interference management.}

Interference management is critical in random cache based networks. This is because when the nearest  BS does not cache the requested content, the user will be served by a relatively farther BS, which makes the signal usually weaker than the interference. To increase received signal power under random cache, \cite{wen2018random}\cite{chae2017content} jointly considered random cache and cooperative transmission to optimize the STP in HetNets. Nevertheless, \cite{wen2018random}\cite{chae2017content} are studied without considering the multi-antenna. To take advantages of multiple antennas, \cite{kuang2017random} adopted maximal ratio transmission (MRT) beamforming to boost the desired signal in  multi-antenna cache-enabled networks with limited backhaul. Moreover, part of the degree of freedom (DoF) can be utilized to avoid those dominant interference in cache-enabled networks. Adopting coordinated beamforming (CBF), \cite{xu2017analysis} considered the analysis and optimization of STP in multi-antenna cache-enabled networks, where a certain number of small base stations (SBSs) form the coordination cluster. However, none of the aforementioned literatures  has addressed the analysis and optimization of  per-user throughput, which is an important performance metric in cache-enabled networks\cite{li2018survey}.



In addition, the CBF scheme adopted by \cite{xu2017analysis} ignored the ``BS selection conflict problem'' in dynamic clustering, i.e., a joint BS shared by different clusters can be selected by different users who share the same spectrum at the same time.  Although this problem can be solved by allocating  orthogonal time-frequency resources to adjoint clusters\cite{park2016cooperative},  it will result in a reduction in the per-user available bandwidth. Therefore, the scheme in \cite{xu2017analysis} may not be suitable for throughput maximization. Recently, \cite{li2015user,cui2016user} proposed a tractable interference nulling (IN) strategy in traditional networks with connection-based association scheme. \cite{li2015user,cui2016user} used random elimination strategy to treat the BS selection conflict problem, which can avoid the reduction of throughput. However, these IN strategies  cannot be directly applied to cache-enabled HetNets with content-centric association scheme ,  as different association schemes result in different distributions of the locations of serving and interfering BSs, and hence lead to different geographic locations of  the dominant interference. Therefore,  to  fully utilize the storage capacity and achieve promising per-user throughput, further research is required to apply IN to random cache based HetNets and understand the relationship between random cache and IN.


In this paper, we apply the user-centric inter-cell IN strategy to random cache based HetNets and jointly investigate the benefits of random cache and IN. Employing IN, each SBS utilizes zero-forcing beamforming (ZFBF) to avoid the interference to nearby SBS scheduled users. Since the cache-enabled networks aim at efficient content delivery\cite{wen2017cache}, we analyze and optimize the per-user throughput. Our main contributions are summarized as follows:

\begin{itemize}
\item We present a framework for the per-user throughput analysis in random cache based multi-antenna HetNets with user-centric inter-cell IN.
\item {By carefully handling the different types of interfering SBSs and employing appropriate approximations}, we theoretically derive the explicit expressions for the per-user throughput  based on stochastic geometry, which can be easily calculated  through mathematical software.
\item The optimal cache placement is obtained via standard  optimization techniques. The optimization results indicate that more different contents can be cached in the SBS tier when SBSs are equipped with more antennas.
\item By numerical results, we analyze the influences of various system parameters, including physical-layer related parameters and content-layer related parameters. Moreover, we show that the benefits brought by random cache are enhanced by the IN. The results show that joint random cache and IN can achieve a significant gain in the per-user throughput over existing baseline schemes.
\end{itemize}

\section{System Model}
We consider a downlink two-tier HetNets where macro base stations (MBSs) are overlaid with small base stations (SBSs). {The locations of} MBSs and SBSs are modeled as two independent PPPs ${\Phi _0}$ and ${\Phi _1}$ with densities ${\lambda _0}$ and ${\lambda _1}$, respectively. The MBSs and SBSs are equipped with $N_0$ and $N_1$ antennas, respectively. The single antenna users are located according to an independent PPP ${\Phi _u}$ with ${\lambda _u}$, where ${\lambda _u} \gg {\lambda _{\rm{1}}} > {\lambda _{\rm{0}}}$. We assume that all BSs are fully loaded and active as in \cite{cui2016analysis,kuang2017random,wen2018random}. Without loss of generality, we focus on the typical user ${u _0}$ located at the origin.

Let ${\cal F} \buildrel \Delta \over = \left\{ {1,2,...,F} \right\}$ denote the set of $F$ files in the network. {For ease of illustration,  all files are assumed to have the equal unit size as \cite{chen2017probabilistic,wen2017cache,chae2017content,kuang2017random,xu2017analysis}.} We assume that the popularity $a_n$ of the file $n$ is arranged in a descending order according to the law of Zipf {as \cite{chen2017probabilistic,wen2017cache,cui2016analysis,wen2018random,chae2017content,kuang2017random,xu2017analysis}, i.e., ${a_n} = \frac{{{n^{ - \gamma }}}}{{\sum\nolimits_{i \in {\cal F}} {{i^{ - \gamma }}} }}$,} where the Zipf exponent $\gamma$ represents the skewness of the popularity distribution. We assume that each MBS is equipped with no cache but is connected to the core network through backhaul link with high capacity. Each SBS is equipped with a cache of size $C$, where $C \le F$. A random cache strategy is adopted where each SBS independently caches file $n\in {\cal F}$ with probability $T_n \in [0,1]$, and we have $\sum\nolimits_{n \in {\cal F}} {{T_n}}  \le C$ due to the cache storage limit. Hereafter, $ {\mathbf{T}} \buildrel \Delta \over = {\left( {{T_n}} \right)_{n \in {\cal F}}}$ is referred to as ``cache distribution". {Denote ${{\cal F}^+} \buildrel \Delta \over = \left\{ {n|n \in {\cal F},{T_n} > 0} \right\}$, hence $\left| {{{\cal F}^+}} \right|$ represents the number of different files could be stored in the SBS tier.  }

The user who requests file $n$ is associated with the nearest SBS cached file $n$ within $R_c$, where $R_c$ is the distance threshold used to define the maximum service distance for the SBSs. Specifically, if no SBS within $R_c$  stores the requested file, the nearest MBS will retrieve file from the core network and then fulfill the user request. {After the user association, each BS schedules its associated users according to TDMA, i.e., scheduling one user in each time slot. Therefore, there is no intra-cell interference.}

Orthogonal frequencies are applied at MBSs and SBSs to avoid inter-tier interference. The available bandwidths of each MBS and SBS for $u_0$ are $W_m$ and $W_s$, respectively. Similar to \cite{wen2018random}, we assume ${W _s}>{W _m}$. To suppress those dominant interference within the SBS tier, we consider a user-centric inter-cell interference nulling strategy, assuming each SBS uses ZFBF  to avoid interference with neighboring users.  {Considering that in the cache-enabled networks, dominant interferences come from the SBSs that are closer than the serving SBS, thus the user associated with SBS tier will send an interference mitigation request to all the interfering SBSs within distance $R_c$.}    Each SBS can handle at most $N_1-1$ requests due to spatial DoF constraint. Let $K_x$ represent the number of requests received by SBS $x$, if ${K_x} \ge {N_1}$, we assume that SBS $x$ will randomly choose $N_1-1$ users to suppress interference. Therefore, SBS $x$ will use ${\max \left( {{N_1} - {K_{{x}}},1} \right)}$ DoF to boost the desired signal to its scheduled user. Note that if $K_x=0$, SBS $x$ will utilize maximal ratio transmission (MRT) precoding to serve its scheduled user.

In this paper, we consider the interference-limited regime and hence ignore the thermal noise. For notation simplicity, we use number 0 and 1 to distinguish the MBS tier and the SBS tier. We denote the serving BS of $u_0$ associated with tier $z$ as BS $x_{z,0}$, where $z \in \left\{ {0,1} \right\}$.  When the typical user $u_0$ is associated with the SBS $x_{1,0}$, different from the conventional connection-based networks, there are three types of interferers. Let $\Phi _1^{1}$ denote the set of interfering SBSs closer than serving SBS $x_{1,0}$, $\Phi _1^{2}$ denote the set of interfering SBSs farther than $x_{1,0}$ but closer than $R_c$, and $\Phi _1^{3}$ denote the set of interfering SBSs farther than $R_c$. Note that $\Phi _1^{1}$ and $\Phi _1^{2}$ are the set of SBSs who receive the IN request from $u_0$ but cannot mitigate interference for $u_0$ due to the limitation of DoF. Then the received signal of $u_0$ associated with tier $1$ is given by:
\begin{align}
{y_{1,0}} = &{\left| {{x_{1,0}}} \right|^{ - \frac{{{\alpha _1}}}{2}}}{\rm \bf{h}}_{{x_{1,0}}}^ * {{\rm \bf{w}}_{{x_{1,0}}}}{{{s}}_{{x_{1,0}}}} \\\nonumber
&+ \sum\limits_{{x_{1,j}} \in \Phi _1^1 \cup \Phi _1^2 \cup \Phi _1^3} {{{\left| {{x_{1,j}}} \right|}^{ - \frac{{{\alpha _1}}}{2}}}{\rm \bf{h}}_{{x_{1,j}}}^ * {{\rm \bf{w}}_{{x_{1,j}}}}{{{s}}_{{x_{1,j}}}}},
\end{align}
where ${{\rm \bf{h}}_{{x_{}}}} \in {\cal C}{\cal N}\left( {{{\bf 0}_{{N_1} \times 1}},{{\rm \bf I}_{{N_1}}}} \right)$ denotes the small-scale fading between SBS $x$ and $u_0$, $|x|$ denotes the distance from  BS $x$ to $u_0$, ${\left| {{x_{}}} \right|^{ - \frac{{{\alpha _1}}}{2}}}$ represents the large-scale path loss, where $\alpha_1>2$ is the path loss exponent, $s_x$ is the information symbol from SBS $x$ to $u_0$, and ${\rm \bf w}_x$ denotes the beamforming vector for SBS $x$.

 We assume that serving SBS $x_{1,0}$ receives $K_{x_{1,0}}$ IN requests, and  SBS $x_{1,0}$ will handle ${\min \left( { {K_{{x_{1,0}}}},{N_1} -1} \right)}$ IN requests.  Then the ZF beamforming vector ${\rm \bf w}_{x_{1,0}}$ is given by:
\begin{align}
{{\rm \bf {w}}_{{x_{1,0}}}} = \frac{{\left( {{{\rm \bf I}_{{N_1}}} - {\rm \bf H}{{\left( {{{\rm \bf H}^ * }{\rm \bf H}} \right)}^{ - 1}}{{\rm \bf H}^ * }} \right){{\rm \bf {h}}_{{x_{1,0}}}}}}{{\left\| {\left( {{{\rm \bf I}_{{N_1}}} - {\rm \bf H}{{\left( {{{\rm \bf H}^ * }{\rm \bf H}} \right)}^{ - 1}}{{\rm \bf H}^ * }} \right){{\rm \bf {h}}_{{x_{1,0}}}}} \right\|}},
\end{align}
where  ${\rm \bf H}{\rm{ = }}\left[ {{{\rm \bf {h}}_1}, \ldots ,{{\rm \bf {h}}_{\min \left( {{K_{{x_{1,0}}}},{N_1} - 1} \right)}}} \right]$ is the channels between SBS  $x_{1,0}$ and ${\min \left( { {K_{{x_{1,0}}}},{N_1} -1} \right)}$ selected users, and $\rm \bf H^*$ is the conjugate transpose of $\rm \bf H$.

However, if no SBS within $R_c$  stores the requested file, $u_0$ will be associated with tier $0$ and served by the nearest MBS $x_{0,0}$ that utilizes MRT precoding.  In this case, the form of received signal is similar to that in \cite{kuang2017random}. Therefore, the signal-to-interference ratio (SIR) of $u_0$ associated with tier $z$ is ${{\rm SIR}_z} = \frac{{{S_z}}}{{{I_z}}}$, $z \in \left\{ {0,1} \right\}$, i.e.,
\begin{align}\label{newEq1}
&{{\rm SIR}_{\rm{0}}} = \frac{{{g_{{x_{{\rm{0}},0}}}}{{\left| {{x_{{\rm{0}},0}}} \right|}^{ - {\alpha _{\rm{0}}}}}}}{{{{\sum\nolimits_{{x_{0,j}} \in {\Phi _0}\backslash {x_{0,0}}} {{g_{{x_{0,j}}}}{{\left| {{x_{0,j}}} \right|}^{ - {\alpha _0}}}} }_{}}}},\\\label{111}
&{{\rm SIR}_{\rm{1}}} = \frac{{{g_{{x_{{\rm{1}},0}}}}{{\left| {{x_{{\rm{1}},0}}} \right|}^{ - {\alpha _{\rm{1}}}}}}}{{\sum\nolimits_{{x_{1,j}} \in \Phi _1^1 \cup \Phi _1^2 \cup \Phi _1^3} {{g_{{x_{1,j}}}}{{\left| {{x_{1,j}}} \right|}^{ - {\alpha _1}}}} }}{\rm{.}}
\end{align}
Here, ${{g_{{x_{z,j}}}}},z\in\left\{0,1\right\}$ is the equivalent channel gain from BS $x_{z,j}$ to user $u_0$, including channel coefficients and the beamforming. $\alpha_z$ is the path loss exponent.  According to \cite{li2015user}\cite{cui2016user}, the information signal channel gain ${{g_{{x_{z,0}}}}}$ follows Gamma distributed, i.e., ${g_{{x_{0,0}}}} \sim {\rm{Gamma}}\left( {{N_0},1} \right)$ and ${g_{{x_{1,0}}}} \!\sim\! {\rm Gamma}\left( {\max \left( {{N_1} - {K_{{x_{1,0}}}},1} \right),1} \right)$, and the interfering channel gain ${{g_{{x_{z,j}}}}},z\!\in\!\left\{0,1\right\}$ is exponential distributed with mean 1.




\section{Analysis of Per-user Throughput}
In this section, we will derive the expression of the per-user throughput under a given cache distribution. The per-user throughput of the typical user is given by
\begin{align}\label{defination_throughput}
R({\bf {T}}) = \sum\nolimits_{n \in {\cal F}} {{a_n}\left( {{W_m}{{{{P}}_m^n}}R_0^n + {W_s}{{{P}_s^n}}R_1^n} \right)} ,
\end{align}
where  $R_0^n$ and $R_1^n$ denote the average spectral efficiency when typical user requesting file $n$ is served by MBS and SBS, respectively, ${{{P}_m^n}}$ refers to the probability that typical user requesting file $n$ is served by MBS, and ${{P}_s^n}$ denotes the ``cache hit probability", i.e., there is an SBS caching the requested file $n$ and within cooperation region (i.e., within $R_c$).

Before calculating throughput, we will first characterize the distribution of  the crucial parameter $K_{x}$. To make the analysis tractable, we make the following approximations:
\begin{appro}\label{appro1}
The scheduled micro users form a  homogeneous PPP  ${\Phi _u^1}$ with density $\lambda_1$.
\end{appro}
\begin{appro}\label{appro2}
The numbers of IN requests received by different SBSs are independent.
\end{appro}

Similar approximations have been adopted in traditional connection-based networks as \cite{li2015user,cui2016user,chen2018stochastic},  and their accuracy in content-centric networks will be verified in Section \ref{numberical results}. With Approximation \ref{appro1}, since SBS $x$ will receive the IN requests from scheduled users at distance $r$ away from SBS $x$ if $r<R_c$,  $K_{x}$ is a Poisson distributed variable with parameter $\overline{K}=\pi  R_c^2\lambda_1$. Then, we can obtain the probability mass function of $K_{x}$, i.e., ${p_K}(k) = \frac{{{{\left( {\pi {R_c}^2{\lambda _1}} \right)}^k}}}{{k!}}{e^{ - \pi {R_c}^2{\lambda _1}}}$.

Based on $p_K(k)$, the probability that a SBS received the request from $u_0$ but can not eliminate the interference for $u_0$, i.e., the interference residual ratio, can be calculated as follows:
\begin{align}\label{newEq6}
\varepsilon  = \sum\nolimits_{k = {N_1} - 1}^\infty  {\frac{{k + 1 - ({N_1} - 1)}}{{k + 1}}{p_K}(k)}.
\end{align}

Based on  the thinning theory of PPP and Approximation \ref{appro2},  when $u_0$ requests file $n$, the densities of three types of interfering SBSs $\Phi _1^{1}$, $\Phi _1^{2}$ and $\Phi _1^{3}$, i.e., $\lambda _1^{1}(x)$, $\lambda _1^{2}(x)$ and $\lambda _1^{3}(x)$ can be obtained as follows:
\begin{align}\label{density}
&\lambda _1^{1}(x) \!=\! {\lambda _1}(1 - {T_n})\varepsilon {\mathbbm{1}}\left( {\left\| x \right\| \in \left[ {0,{x_{1,0}}} \right)} \right),\nonumber\\
&\lambda _1^{2}(x) \!=\! {\lambda _1}\varepsilon {\mathbbm{1}}\left( {\left\| x \right\| \in \left[ {{x_{1,0}},{R_c}} \right)} \right),\nonumber\\
&\lambda _1^{3}(x) \!=\! {\lambda _1}{\mathbbm{1}}\left( {\left\| x \right\| \!\in\! \left[ {{R_c},\infty } \right]} \right),
\end{align}
where $\mathbbm{1}(*)$ denotes the indicator function. Aided by the above results, we derive the per-user throughput based on the tools from stochastic geometry.

\begin{thm}\label{theorem per-user throughput}
In the cache-enabled multi-antenna HetNets with user-centric inter-cell interference nulling strategy, the per-user throughput is given by
\begin{align}\label{per-user throughput}
R({\bf T}) = \sum\nolimits_{n \in {\cal F}} {{a_n}\left( {{f_m}\left( {{T_n}} \right) + {f_s}\left( {{T_n}} \right)} \right)},
\end{align}
with
\begin{small}
\begin{align}\label{newEq81.5}
{f_m}\left( {{T_n}} \right)&={W_m}{e^{ - \pi {T_n}{\lambda _1}{R_c}^2}}\!\!\!\int_0^\infty \!\!\!\! {\frac{{\left( {1 - {{\left( {1 \!+\! t} \right)}^{ - {N_0}}}} \right)}}{{t\!\left(\! {1 \!+\! \frac{2}{{{\alpha _0}}}{t^{\frac{2}{{{\alpha _0}}}}}B\!\left(\!1\! -\! \frac{2}{{{\alpha _0}}},\frac{2}{{{\alpha _0}}},\frac{t}{{1 + t}}\right)} \!\right)\!}}dt},\\
{f_s}\left( {{T_n}} \right) &= {W_s}\pi {\lambda _1}{T_n}\int\nolimits_0^\infty  {\int\nolimits_0^{R_c^2} {} } \frac{{{S^t}({N_1},k)}}{t}\nonumber\\
&\times {\rm{exp}}\!\left(\! {\;{Z^{x,t}}\left( {\frac{1}{{1 + t}},\frac{t}{{1 + t}},\frac{{{x^{\frac{{{\alpha _1}}}{2}}}t}}{{{R_c}^{{\alpha _1}} + {x^{\frac{{{\alpha _1}}}{2}}}t}}} \right)} \!\right)\!dxdt,
\end{align}
\end{small}where
\begin{small}
\begin{align}\label{newEq10}
&{S^{\rm{t}}}({N_1},k) = \left( {1 - \sum\nolimits_{k = 0}^\infty  {{{\left( {1 + t} \right)}^{ - \max({N_1} - k,1)}}{p_K}(k)} } \right),\\
&Z^{x,t}(u,y,z) =- \pi {\lambda _1}x{T_n} \nonumber\\
&\qquad+x\left( {\varepsilon (1 \!-\! {T_n}){C_2^t}(u) +\! \varepsilon {C_1^t}(y) + (1 \!-\! \varepsilon ){C_1^t}(z)} \right),
\end{align}
\end{small}${C_z^t}\left( u \right) =  - \pi {\lambda _z}{t^{\frac{2}{{{\alpha _z}}}}}\frac{2}{{{\alpha _z}}}B\left( {1 - \frac{2}{{{\alpha _z}}},\frac{2}{{{\alpha _z}}},u} \right)$, $z \in \left\{0,1\right\}$, ${C_2^t}\left( u \right) =  - \pi {\lambda _1}{t^{\frac{2}{{{\alpha _1}}}}}\frac{2}{{{\alpha _1}}}B\left( {\frac{2}{{{\alpha _1}}},1 - \frac{2}{{{\alpha _1}}},u} \right)$, and $B(x,y,z) = \int_0^z {{u^{x - 1}}{{\left( {1 - u} \right)}^{y - 1}}du}$ is the incomplete Beta function.
\end{thm}

\begin{proof}
When the typical user $u_0$ requests file $n$, based on the capacity calculation lemma in \cite{hamdi2010useful}, spectral efficiency is given by
\begin{align}\label{SE}
&{R_z^n} = \mathbb {E}\left[ {\ln \left( {1 + \rm {SIR_z}} \right)} \right] \nonumber\\
&=\!\!\int_{\rm{0}}^\infty \!\!\!\! {\int_0^\infty  {\frac{{{{\cal L}_{{I_z}}}\left( {t\left| {\;\left| {{x_{z,0}}} \right|} \right.} \right)\left( {1 \!-\! {{\cal L}_{{S_z}}}\left( {t\left| {\;\left| {{x_{z,0}}} \right|} \right.} \right)} \right)}}{t}{f_{{x_{z,0}}}}\left( x \right)dtdx} },
\end{align}
where $z \in \left\{ {0,1} \right\}$, and ${\cal L}_{X}(t)$ is the Laplace transform of $X$. Based on (\ref{SE}), we calculate $R_0^n$ and $R_1^n$, respectively.

To begin with, based on the void probability of PPP, the probability density function (PDF) of the service distances in two layers can be calculated as follows:
\begin{align}\label{fsfm}
&{f_{{x_{0,0}}}}(x) = 2\pi {\lambda _0}x\;{\exp\left({ - \pi {\lambda _0}{x^2}}\right)},\nonumber\\
&{f_{{x_{1,0}}}}(x\left| {{x_{1,0}} < {R_c}} \right.) = \frac{{2\pi {\lambda _1}{T_n}x\exp \left( { - \pi {T_n}{\lambda _1}{x^2}} \right)}}{{1 - \exp \left( { - \pi {T_n}{\lambda _1}{R_c}^2} \right)}},
\end{align}
and the cache hit probability ${{P}_s^n}$ can be derived as follows:
\begin{align}\label{PsPm}
{{{P}_s^n}} \!=\!1 \!- \!{{{P}_m^n}}\!=\! \int_0^{{R_c}}\! {2\pi {T_n}} {\lambda _1}r{e^{ - \pi {T_n}{\lambda _1}{r^2}}}dr \!= \!1\! - \!{e^{ - \pi {T_n}{\lambda _1}{R_c}^2}}.
\end{align}

To obtain $R_0^n$, we first calculate  the Laplace transforms of $S_0$ and $I_0$, respectively. We rewrite the expressions of $S_0$ and $I_0$ in (\ref{newEq1}) as ${S_0} = {g_{{x_{0,0}}}}$ and ${I_0}{\rm{ = }}\sum\nolimits_{{x_{0,j}} \in {\Phi _0}\backslash {x_{0,0}}} {{g_{{x_{0,j}}}}{{\left| {{x_{0,j}}} \right|}^{ - {\alpha _0}}}{{\left| {{x_{0,0}}} \right|}^{{\alpha _0}}}} $, respectively.

As $g_{x_{0,0}}$ is Gamma distributed, we obtain ${{\cal L}_{{S_0}}}(t|\left| {{x_{0,0}}} \right|) $ as follows:
\begin{align}\label{R0_1}
\begin{array}{l}
{{\cal L}_{{S_0}}}(t|\left| {{x_{0,0}}} \right|)  = {{\mathbb {E}}_{{g_{{x_{0,0}}}}}}\left[ {{e^{ - {g_{{x_{0,0}}}}t}}} \right]= {\left( {1 + t} \right)^{ - {N_0}}}.
\end{array}
\end{align}

By applying the probability generating function (PGFL) of PPP, we derive the Laplace transform of $I_0$ as follows:
\begin{align}\label{R0_2}
\begin{array}{l}
{{\cal L}_{{I_0}}}(t\left| {\left| {{x_{0,0}}} \right|} \right.) = \exp \left( {{C_0^t}\left( {\frac{t}{{1 + t}}} \right){{\left| {{x_{0,0}}} \right|}^2}} \right).
\end{array}
\end{align}

Then, by substituting (\ref{fsfm}), (\ref{R0_1}) and (\ref{R0_2}) into (\ref{SE}), we can obtain $R_0^n$ after some algebraic manipulations.

To obtain $R_1^n$, we first calculate the Laplace transforms of $S_1$ and $I_1$, respectively. We rewrite the expressions of $S_1$ and $I_1$ in (\ref{111}) as ${S_1} = {g_{{x_{1,0}}}}$ and ${I_1}{\rm{ = }}\sum\nolimits_{{x_{1,j}} \in \Phi _1^1 \cup \Phi _1^2 \cup \Phi _1^3} {{g_{{x_{1,j}}}}{{\left| {{x_{1,j}}} \right|}^{ - {\alpha _1}}}{{\left| {{x_{1,0}}} \right|}^{{\alpha _1}}}} $, respectively. Based on the law of total expectation, the Laplace transform of $S_1$ is given by
\begin{align}\label{S1}
{{\cal L}_{{S_1}}}\left( {t\left| {\left| {{x_{1,0}}} \right|} \right.} \right) = \sum\nolimits_{k=0}^\infty  {{{\left( {1 + t} \right)}^{ - \max({N_1} - k,1)}}} {p_K}(k).
\end{align}

In order to calculate the Laplace transform of $I_1$, we define ${I_1} = I_1^1 + I_1^2 + I_1^3$, where $I_1^z{\rm{ = }}\sum\nolimits_{{x_{1,j}} \in \Phi _1^z} {{{\left| {{x_{1,0}}} \right|}^{{\alpha _1}}}{{\left| {{x_{1,j}}} \right|}^{ - {\alpha _1}}}{g_{{x_{1,j}}}}} $, and $z \in \left\{1,2,3\right\}$. First, using the PGFL of PPP with  densities (\ref{density}) in different regions, we calculate the Laplace transforms of three types of interferences as follows:
\begin{align}
\begin{array}{l}
{{\cal L}_{I_{\rm{1}}^{\rm{1}}}}\left( {t\left| {\left| {{x_{1,0}}} \right|} \right.} \right)\\
 = \exp \left( { - 2\pi \varepsilon (1 - {T_n}){\lambda _1}\int\limits_0^{\left| {{x_{1,0}}} \right|} {\left( {1 - \frac{1}{{1 + {{\left| {{x_{1,0}}} \right|}^{{\alpha _{\rm{1}}}}}{r^{ - {\alpha _{\rm{1}}}}}t}}} \right)} rdr} \right)\\
 = \exp \left( {\varepsilon (1 - {T_n})C_{\rm{2}}^{\rm{t}}\left( {\frac{1}{{1 + t}}} \right){{\left| {{x_{1,0}}} \right|}^2}} \right),\\
{{\cal L}_{I_{\rm{1}}^{\rm{2}}}}\left( {t\left| {\left| {{x_{1,0}}} \right|} \right.} \right)\! =\! \exp \!\left(\!\! { - 2\pi \varepsilon {\lambda _1}\!\int\limits_{\left| {{x_{1,0}}} \right|}^{{R_c}} \!{\left( {1 \!-\! \frac{1}{{1 + {{\left| {{x_{1,0}}} \right|}^{{\alpha _{\rm{1}}}}}{r^{ - {\alpha _{\rm{1}}}}}t}}} \right)} rdr}\!\! \right)\\
 = \exp \left( {\varepsilon \left( {C_1^t\left( {\frac{t}{{1 + t}}} \right) - C_1^t\left( {\frac{{{{\left| {{x_{1,0}}} \right|}^{{\alpha _1}}}{R_c}^{ - {\alpha _1}}t}}{{1 + {{\left| {{x_{1,0}}} \right|}^{{\alpha _1}}}{R_c}^{ - {\alpha _1}}t}}} \right)} \right)\;{{\left| {{x_{1,0}}} \right|}^2}} \right),\\
{{\cal L}_{I_{\rm{1}}^{\rm{3}}}}\left( {t\left| {\left| {{x_{1,0}}} \right|} \right.} \right) = \!\exp \left( \!{ - 2\pi {\lambda _1}\int\limits_{{R_c}}^\infty  {\left( {1 - \frac{1}{{1 + {{\left| {{x_{1,0}}} \right|}^{{\alpha _{\rm{1}}}}}{r^{ - {\alpha _{\rm{1}}}}}t}}} \right)} rdr} \!\!\right)\\
 = \exp \left( {C_1^t\left( {\frac{{{{\left| {{x_{1,0}}} \right|}^{{\alpha _1}}}{R_c}^{ - {\alpha _1}}t}}{{1 + {{\left| {{x_{1,0}}} \right|}^{{\alpha _1}}}{R_c}^{ - {\alpha _1}}t}}} \right){{\left| {{x_{1,0}}} \right|}^2}} \right).
\end{array}
\end{align}

Then, we can calculate the Laplace transform of $I_1$ as ${{\cal L}_{I_1^{}}}\left( {t\left| {\left| {{x_{1,0}}} \right|} \right.} \right) = {{\cal L}_{I_1^1}}\left( {t\left| {\left| {{x_{1,0}}} \right|} \right.} \right){{\cal L}_{I_1^2}}\left( {t\left| {\left| {{x_{1,0}}} \right|} \right.} \right){{\cal L}_{I_1^3}}\left( {t\left| {\left| {{x_{1,0}}} \right|} \right.} \right)$. We omit the details here due to the space limitation. By substituting ${{\cal L}_{I_1^{}}}\left( {t\left| {\left| {{x_{1,0}}} \right|} \right.} \right)$, (\ref{fsfm}), and (\ref{S1}) into (\ref{SE}), we can obtain $R_1^n$ after some algebraic manipulations.

Finally, by substituting $R_0^n$, $R_1^n$ and (\ref{PsPm}) into (\ref{defination_throughput}), we can obtain the per-user throughput.
\end{proof}


Note that if each SBS adopts MRT precoding instead of ZFBF, the per-user throughput can be obtain easily by Theorem \ref{theorem per-user throughput}. Moreover, our analytical result (\ref{per-user throughput}) does not contain the high order derivatives of the Laplace transform as in \cite{xu2017analysis} and the matrix inversion as in \cite{kuang2017random,li2015user}, thus it can be easily calculated by mathematical software. Under the parameters in Section \ref{numberical results}, the computation time using Monte-Carlo simulations is more than 500 times of that using our analytical results, which demonstrates that our results are more efficient and tractable than Monte-Carlo simulations.

\section{Throughput Maximization}
In this section, we design the caching placement $\bf T$ in our scheme by solving the following optimization problem:
\begin{subequations}\label{optimal1}
\begin{align}
&{\bf{P0}}: \;\;\mathop {\max }\limits_{\mathbf T} \;\;R\left( {\mathbf T} \right){\rm{ = }}\sum\nolimits_{n \in {\cal F}} {{a_n}\left( {{f_m}\left( {{T_n}} \right) + {f_s}\left( {{T_n}} \right)} \right)} \\\label{222}
&\;\;\;\;\;\;\;\;\;\;\;{\rm{s.t.}}\;\;\;\;\sum\nolimits_{n \in {\cal F}} {{T_n}}  \le C,\\\label{cons1}
&\;\;\;\;\;\;\;\;\;\;\;\;\;\;\;\;\;\;\;\;\;0 \le {T_n} \le 1,n \in {\cal F}.
\end{align}
\end{subequations}

Since storing more files will increase the per-user throughput, without loss of optimality, we rewrite the constraint (\ref{222}) as follows:
\begin{align}\label{cons2}
\sum\nolimits_{n \in {\cal F}} {{T_n}}  = C.
\end{align}

Notice that it is difficult to ensure the convexity of $ R\left( {\mathbf T} \right)$ in general, due to the summation of two tiers and the complex structure of $f_s(T_n)$. However, problem ${\bf{P0}}$ is a continuous optimization of a differentiable function over a convex set, since function $ R\left( {\mathbf T} \right)$ is differentiable and the constraints in (\ref{cons1}) (\ref{cons2}) are linear. Therefore, we can use the gradient projection method in \cite{cui2016analysis} to compute a local-optimal solution (local-Opt.).

To obtain some design insights, we first consider a special case with $W_s\gg W_m$, since more users will be associated with an MBS than an SBS in the traditional HetNets, which makes the available bandwidth of an SBS for $u_0$ is larger than that of an MBS in general\cite{wen2018random}. In this case, since the impact of the throughput contributed by MBS tier can be safely ignored, we consider the throughput maximization problem as follows:
\begin{align}
&{\bf{P1}}:\;\;\mathop {\max }\limits_{\mathbf T} \;\;\;\underline{R}\left( {\mathbf T} \right){\rm{ = }}\sum\nolimits_{n \in {\cal F}} {{a_n}{f_s}\left( {{T_n}} \right)} \nonumber\\
&\;\;\;\;\;\;\;\;\;\;\;{\rm {s.t.}}\;\;\;\;\sum\nolimits_{n \in {\cal F}} {{T_n}}  = C,\;\;\;0 \le {T_n} \le 1,n \in {\cal F}.
\end{align}

However, it is still difficult to ensure the convexity of $\underline{R}\left( {\mathbf T} \right)$  in general case due to the complex form of $f_s(T_n)$. To further simplify $\bf{P1}$, we  consider a special case that ${R_c} \le \sqrt {\frac{{\rm{2}}}{{\pi {\lambda _1}}}} $ and $N_1>4$. Note that this means that the average number of IN requests received by a SBS is less than 2, which is close to the best value as shown in \cite{li2015user}. Meanwhile, in a realistic scenario, the size of the coordinated set would not be large due to the signaling overhead. In this case, due to $\varepsilon  \to {\rm{0}}$, the second order derivative of the objective function ${{\underline {R}}({\bf{T}})}$ with respect to $T_n$ can be expressed as:
\begin{align}
\begin{array}{l}
\frac{{{\partial ^2}{{{\underline {R}}({\bf{T}})}}}}{{T_n^2}}{\rm{ = }}{a_n}{W_s}{\pi ^2}\lambda _1^2\int\nolimits_0^\infty \!{\int\nolimits_0^{R_c^2} {} } \frac{{{S^t}({N_1},k)}}{t}\left( {{T_n}\pi {\lambda _1}x - 2} \right)\\
\times {\rm{exp}}\left( {{Z^{x,t}}\left( {\frac{1}{{1 + t}},\frac{t}{{1 + t}},\frac{{{x^{\frac{{{\alpha _1}}}{2}}}t}}{{{R_c}^{{\alpha _1}} + {x^{\frac{{{\alpha _1}}}{2}}}t}}} \right)} \right)xdxdt.
\end{array}
\end{align}

It is easy to verify that $\frac{{{\partial ^2}{{\underline {R}}({\bf{T}})}}}{{T_n^2}} \le 0$ when ${R_c} \le \sqrt {\frac{{\rm{2}}}{{\pi {\lambda _1}}}} $, thus ${\bf{P1}}$ is convex. Therefore, we can obtain the near-optimal solutions (near-Opt.) to ${\bf{P0}}$  by solving ${\bf{P1}}$, since the optimal objective value obtained from $\bf{P1}$ in general serves as a tight lower bound of that of problem $\bf{P0}$. Using  Karush-Kuhn-Tucker (KKT) condition, the near-optimal solutions to ${\bf{P0}}$ are given by
\begin{align}\label{kkt}
{T_n}^*\left( {{\nu ^*}} \right) = \left[ {\xi \left( {{\nu ^*}} \right)} \right]_0^1,\;\;\;n \in {\cal F},
\end{align}
where $\left[ x \right]_0^1 \buildrel \Delta \over = \max \{ \min\{ x,1\} ,0\} $, and ${\xi \left( {{\nu ^*}} \right)}$ is the solution over $T_n$ of the equation ${a_n}{f_s}^\prime \left( {{T_n^*}} \right) = {\nu ^*}$. The optimal Lagrange multiplier ${{\nu ^*}}$ can be obtained by bisection search based on the condition $\sum\nolimits_{n \in {\cal F}} {{T_n}^*\left( {{\nu ^*}} \right)}  = C$.

Based on (\ref{kkt}), since ${f_s}^\prime(T_n)$ is a decreasing function and we have ${a_1} \ge {a_2} \ge  \cdots  \ge {a_F} \ge 0$, we can easily obtain that $1 \ge {T_1}^* \ge {T_2}^* \ge  \cdots  \ge {T_F}^* \ge 0$. This indicates that SBSs tend to cache popular files, which is consistent with our intuition.

\vspace{-5pt}
\section{Numerical Results}\label{numberical results}


In this section, numerical results are presented to compare our proposed scheme, i.e., RC combined with IN (RC\&IN), with some existing baseline schemes. Unless otherwise noted, our simulation is based on the following setting: $W_m=0.2$ MHz, $W_s=5$ MHz, $N_0=10$, $N_1=7$, ${\lambda _u} = \frac{1}{{{{\rm{5}}^{\rm{2}}}\pi }}{\rm {m}^{ - 2}} = {10^2}{\lambda _1} = {10^4}{\lambda _0}$, ${\alpha _0} = {\alpha _1} = 4$, $F=100$, $C=20$, $R_c=100 \rm{m}$, and the Zipf exponent $\gamma=0.5$. We obtain the Monte-Carlo results by averaging over $10^4$ random realization.

\begin{figure}[!t]
\setlength{\abovecaptionskip}{-5pt}
\setlength{\belowcaptionskip}{-10pt}
\centering
\includegraphics[width= 0.4\textwidth]{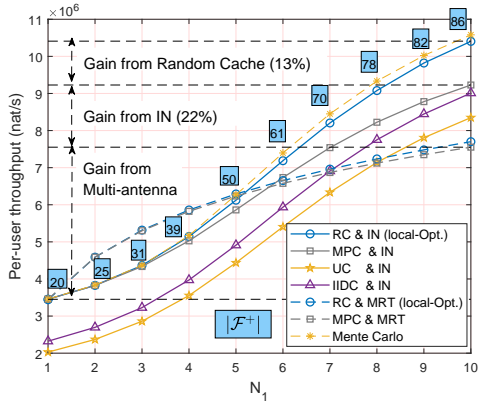}
\DeclareGraphicsExtensions.
\caption{Per-user throughput versus number of antennas at $C=20$, $\gamma=0.5$, and $R_c=100$.}
\label{figure1}

\vspace{-10pt}
\end{figure}

Fig. \ref{figure1} shows the impact of physical-layer related parameter $N_1$ on caching performance, it also verifies our throughput expression with Monte-Carlo simulations. Fig. \ref{figure1} demonstrates that without IN, the file diversity gain (FDG) (the gap between RC\&MRT and MPC\&MRT) is negligible due to the overwhelming interference. In contrast, FDG (the gap between RC\&IN and MPC\&IN) is obvious with IN.  Fig. \ref{figure1} also shows that both the FDG (the gap between RC\&IN and MPC\&IN) and cooperation gain (the gap between RC\&IN and RC\&MRT) benefit from IN, especially at large $N_1$. This is because increasing $N_1$ enhances the ability of IN, reduces interference residual ratio $\varepsilon$, and protects the performance when users download less popular files from farther SBSs. Particularly, we can observe that when $N_1$ is relatively small or $R_c$ is relatively large, it is better to switch to the non-coordination case, i.e., adopting MRT precoding to enhance its desired signal, because of the small effective channel gain and the large $\varepsilon$ in these regions. Moreover, it is observed that more different contents can be cached in the SBS tier when SBSs are equipped with more antennas.

\begin{figure}[!t]
\setlength{\abovecaptionskip}{-5pt}
\setlength{\belowcaptionskip}{-10pt}
\centering
\includegraphics[width= 0.4\textwidth]{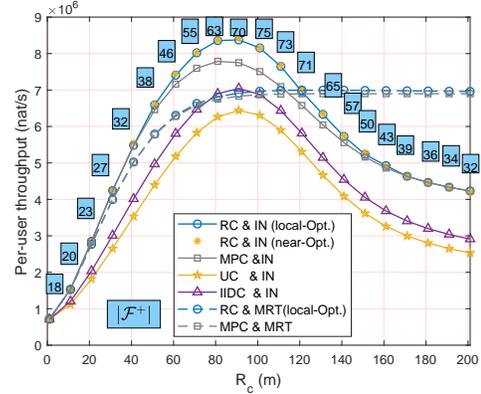}
\DeclareGraphicsExtensions.
\caption{Per-user throughput versus  cooperation radius at $C=20$, $\gamma=0.5$ and $N_1=7$.}
\label{figure2}

\vspace{-11.5pt}
\end{figure}
Fig. \ref{figure2} shows the impact of physical-layer related parameter $R_c$ on caching performance. It also demonstrates that our near-optimal solutions are very close to our local-optimal solutions, even at large $R_c$ region. It also shows that the optimal $R_c$ is around $ \sqrt {2/\left( {\pi {\lambda _1}} \right)}  $. In addition, it is observed that our proposed design first increases and then decreases when $R_c$ increases. The reason lies in that  increasing $R_c$ can first grow the cache hit probability while eliminating more interference. However, when $R_c$ becomes  large enough, the SBS  only has a small DoF for its own signal links, and the interference residual ratio is high. Moreover, we can observe that RC will reduce to MPC when $R_c$ is small or sufficiently large. This is because when $R_c$ is small, the MPC scheme can bring the largest cache hit probability. When $R_c$ is sufficiently large, the cache hit probability and $\varepsilon$ is very large, thus users tend to be associated with the nearest SBS to avoid the heavy interference. Furthermore, it is observed that  more different files cached (i.e., larger $|\cal{F}^+|$) leads to  larger FDG (the gap between RC\&IN and MPC\&IN).

\begin{figure}[!t]
\setlength{\abovecaptionskip}{-5pt}
\setlength{\belowcaptionskip}{-10pt}
\centering
\includegraphics[width= 0.35\textwidth]{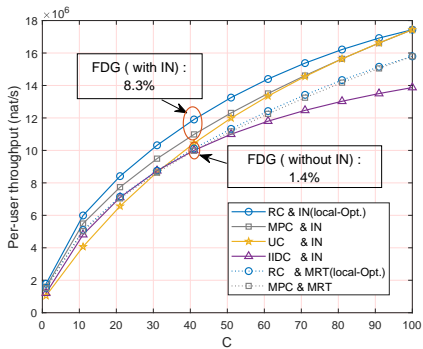}
\DeclareGraphicsExtensions.
\caption{Per-user throughput versus cache size at  $R_c=100$ and $\gamma=0.5$.}
\label{figure3}

\vspace{-10pt}
\end{figure}

\begin{figure}[!t]
\setlength{\abovecaptionskip}{-5pt}
\setlength{\belowcaptionskip}{-10pt}
\centering
\includegraphics[width= 0.35\textwidth]{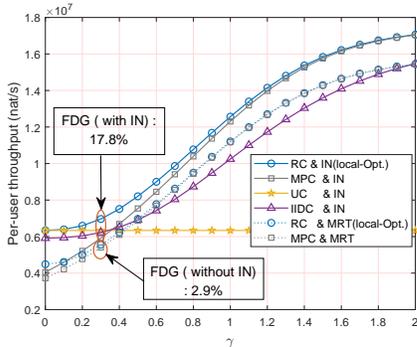}
\DeclareGraphicsExtensions.
\caption{Per-user throughput versus Zipf exponent at $R_c=100$ and $C=20$.}
\label{figure4}

\vspace{-10pt}
\end{figure}
Fig. \ref{figure3} and Fig. \ref{figure4} show the impact of content-layer related parameters on caching performance. It is observed that RC\&IN can fully harvest the FDG compared to RC without IN, since it effectively mitigates the dominant interferences. It is also observed that our proposed design is superior to all the contrastive schemes. Particularly,  Fig. \ref{figure3} shows that cooperation gain (the gap between RC\&IN and RC\&MRT) and throughput of all the schemes increase with cache size. This is because large cache size increases the cache hit probability, more files can be downloaded from the nearest SBSs, thus increasing the throughput. Meanwhile, to gain FDG, large cache size also increases the cache probabilities of those less popular files, thereby increasing the benefits from IN.  Fig. \ref{figure4} shows that optimal cache placement changes from UC to MPC with increasing $\gamma$. This is because large $\gamma$ means that the content requirements become more concentrated, thus SBSs only need to cache the most popular files.

\begin{figure}[!t]
\setlength{\abovecaptionskip}{-5pt}
\setlength{\belowcaptionskip}{-10pt}
\centering
\includegraphics[width= 0.35\textwidth]{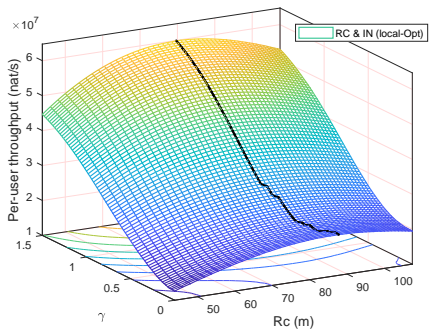}
\DeclareGraphicsExtensions.
\caption{ Per-user throughput under RC\&IN scheme.}
\label{figure5}

\vspace{-10pt}
\end{figure}
Fig. \ref{figure5} shows the relationship between content-layer related parameter $\gamma$ and physical-layer related parameter $R_c$. The black line represents the maximum value of throughput versus  $\gamma$. It is observed that optimal $R_c$ decreases when $\gamma$ increases, which means that UC scheme need larger cooperation region than MPC scheme. The reason lies in that under a more concentrated file request, users may find their desired files in a closer SBS, thus decreasing the optimal cooperation radius.

%
%
%

\section{Conclusion}

In this paper, we jointly considered random cache and interference nulling in multi-antenna HetNets. The explicit expression of the throughput was first obtained by using tools from stochastic geometry. Then, we tackled the throughput maximization problem under cache distribution. Numerical results showed that joint random cache and IN can  sufficiently reap the file diversity gain and  achieve a significant gain in the per-user throughput over existing baseline schemes.



\bibliographystyle{IEEEtran}


\bibliography{myref}

\end{document}